\documentclass[runningheads]{llncs}
\usepackage{amsmath}

\title{Iterative Packing for Demand and Hypergraph Matching}
\author{Ojas Parekh}
\institute{Sandia National Laboratories\thanks{Sandia National Laboratories is a multi-program laboratory managed and operated by Sandia Corporation, a wholly owned subsidiary of Lockheed Martin Corporation, for the U.S. Department of Energy's National Nuclear Security Administration under contract DE-AC04-94AL85000.}, MS 1316, Albuquerque NM 87185, USA
\email{odparek@sandia.gov}}

\begin{document}
\maketitle

\begin{abstract}
Iterative rounding has enjoyed tremendous success in elegantly
resolving open questions regarding the approximability of problems
dominated by covering constraints. Although iterative rounding methods
have been applied to packing problems, no single method has emerged
that matches the effectiveness and simplicity afforded by the covering
case. We offer a  simple iterative packing technique that retains
features of Jain's seminal approach, including the property that the
magnitude of the fractional value of the element rounded during each
iteration has a direct impact on the approximation guarantee. We apply
iterative packing to generalized matching problems including demand matching
and $k$-column-sparse column-restricted packing ($k$-CS-PIP) and obtain approximation
algorithms that essentially settle the integrality gap for these problems.  We present
a simple deterministic $2k$-approximation for $k$-CS-PIP, where an $8k$-approximation was the best deterministic algorithm previously known.  The integrality gap in this case is at least $2(k-1+1/k)$.  We also give a deterministic $3$-approximation for a generalization of demand matching, settling its natural integrality gap.
\end{abstract}

\section{Introduction}

The (maximum weight) matching problem is cornerstone combinatorial optimization problem that has been studied for over 40 years.  The problem is succinctly stated as seeking a maximum weight collection of non-intersecting edges in a weighted graph.  Matching problems have enjoyed continual interest over the years and have been generalized in several orthogonal directions.  The \emph{$k$-hypergraph matching} problem, which is also known as \emph{$k$-set packing}, seeks a maximum weight collection of non-intersecting hyperedges in a weighted hypergraph, with the additional restriction that each hyperedge contain at most $k$ vertices.  Thus 2-hypergraph matching is precisely the matching problem.  While matching is in P, $k$-hypergraph matching is NP-complete for $k > 2$.

Another direction in which matching has recently been generalized is through the augmentation of demands.  The \emph{demand matching} problem, introduced by Shepherd and Vetta~\cite{ShepVet-MOR07}, is defined on a weighted graph $G$ that possesses a demand, $d_e \in \bbbz_+$ for each edge $e$, and a capacity $b_v \in \bbbz_+$ for each vertex $v$.  We seek a maximum weight collection of edges $M$ such that for any vertex $v$, the sum of the demands of the edges in $M$ incident upon $v$ is at most $b_v$ (i.e. $\sum_{e \in M \cap \delta(v)} d_e \leq b_v$).  Demand matching is a common generalization of the matching and knapsack problems: If $d_e = b_v = 1$ for all $e$ and $v$, we recover matching, and by taking $G$ to be a star we may model knapsack.  Demand matching is MAXSNP-complete, even in the uniform weight case~\cite{ShepVet-MOR07}.  Demands are a powerful feature that allow for richer modeling; however, in general the demand version of combinatorial optimization problem can be significantly more difficult to approximate than its unit demand counterpart~\cite{ChekuriEtAl-TALG07,ChekuriEtAl-APPROX09}.

\paragraph{\textbf{Problem definition.}} We consider the \emph{$k$-hypergraph demand matching} ($k$-HDM) problem, which is the natural common generalization of $k$-hypergraph matching and demand matching.  More formally, given a weighted hypergraph $H=(V,\mathcal{E})$ endowed with a demand $d_S \in \bbbz_+$ for each (hyper)edge $S \in \mathcal{E}$ and a capacity $b_v \in \bbbz_+$ for $v \in V$, the problem may be defined by the following integer program ($k$-HDM):
\begin{alignat*}{3}
\text{Maximize} &\quad& \sum_{S \in \mathcal{E}} c_S x_S &\\
  \text{subject to} &&
  \sum_{S | v \in S } d_S x_S &\leq b_v &\quad& \forall v \in V\\
  && x_S & \in \{0,1\} && \forall S \in \mathcal{E}\enspace,
\end{alignat*}
where $|S| \leq k$ for each edge $S \in \mathcal{E}$.  Note the latter restriction yields a constraint matrix with at most $k$ nonzeros per column.  This problem is also known as the \emph{$k$-column-sparse column-restricted packing integer program} ($k$-CS-CPIP) problem.  It is a specialization of the \emph{$k$-column-sparse packing integer program} ($k$-CS-PIP) problem, in which we allow each $S \in \mathcal{E}$ to have different demand values $d^S_v$ at each vertex $v$.

We make the assumption that for each edge $S$, $d_S \leq b_v$ for all $v \in S$.  This so-called \emph{no-clipping} assumption is easy to satisfy by deleting edges that violate it; however, this assumption is necessary in order for the natural LP relaxation to have a bounded integrality gap.  We note that the restriction that $x_S \in \{0,1\}$ is for the sake of exposition and that our results may be extended to apply to multiple copies of edges.

\paragraph{\textbf{Results.}}  Singh and Lau~\cite{SinghLau-STOC07} were the first to extend Jain's celebrated iterative rounding technique~\cite{Jain-Combinatorica01} to address packing constraints.  Their approach obtains an approximate solution that marginally violates the packing constraints by iteratively removing packing constraints involving only a small number of variables.  They were able to apply this elegant idea to resolve an open question concerning minimum cost degree-bounded spanning trees.  More recently, Chan and Lau~\cite{ChanLau-SODA10} employed an interesting combination of an iterative approach and the fractional local ratio method~\cite{BarYehudaEtAl-SICOMP06} to give the first approximation algorithm for the $k$-hypergraph matching problem that matched the integrality gap of the natural LP formulation, which had previousy been established as $k-1+1/k$~\cite{KahnEtAl-Combinatorica93}.  

Our main insight, which differentiates our approach from previous ones, is to iteratively maintain a sparse approximate convex decomposition of the current fractional solution.  This affords us a simple pseudo-greedy technique called iterative packing that yields improved approximation algorithms for $k$-HDM ($k$-CS-CPIP) and special cases that essentially settle the integrality gap of the natural LP formulations.  For instance, iterative packing is able to establish an integrality gap of $k-1+1/k$ for not just $k$-hypergraph matching but for $k$-hypergraph $b$-matching as well.

Akin to Jain's iterative rounding method for covering problems~\cite{Jain-Combinatorica01}, iterative packing is able to leverage large fractional edges to obtain stronger approximation guarantees.  As mentioned above, iterative packing produces and maintains a sparse approximate convex decomposition rather than a single solution, which is likely to have additional applications.  We motivate the technique on the standard matching problem in the next section.

Our first result is a deterministic $2k$-approximation for $k$-HDM ($k$-CS-CPIP) based on the natural LP relaxation.  The integrality gap of this relaxation is at least $2(k-1+1/k)$ (see Sect.~\ref{sec:k-HDM}), hence our result essentially closes the gap.  Prior to our work, deterministic $8k$-approximations~\cite{ChekuriEtAl-APPROX09,BansalEtAl-IPCO10} and a randomized $(ek+o(k))$-approximation~\cite{BansalEtAl-IPCO10} were the best known.  Moreover, even the special case of $k$-hypergraph matching cannot be approximated within a factor of  $\Omega(\frac{k}{\log k})$ unless P=NP~\cite{HazanEtAl-CC06}.

With a more refined application of iterative packing, we are able to derive a 3-approximation for $2$-CS-PIP, which generalizes the demand matching problem.  Prior to our work, a deterministic 3.5-approximation and randomized 3.264-approximation for demand matching were given by Shepherd and Vetta~\cite{ShepVet-MOR07}.  Chakrabarty and Pritchard~\cite{ChakPritch-Algorithmica10} recently gave a deterministic 4-approximation and randomized 3.764-approximation for $2$-CS-PIP.  Shepherd and Vetta also established a lower bound of 3 on the integrality gap of the natural LP for demand matching, hence our result settles the integrality gap for both $2$-CS-PIP and demand matching at 3.

\paragraph{\textbf{Related work.}} Chekuri, Mydlarz, and Shepherd~\cite{ChekuriEtAl-TALG07} presented an approximation algorithm with $O(k)$ guarantee for the restricted version of $k$-HDM in which $\max_S d_S \leq \min_v b_v$.  Their result is part of a framework which they developed based on work of Kolliopoulos and Stein~\cite{KolStein-MathProg04} that relates the integrality gap of a demand-endowed packing problem to its unit demand counterpart.  

While Chekuri \emph{et al.}~\cite{ChekuriEtAl-APPROX09} observed an $8k$-approximation for $k$-HDM, a recent flurry of work has also yielded $O(k)$-approximations for the more general $k$-CS-PIP problem.  Pritchard initiated the improvements with an iterative rounding based $2^k k^2$-approximation~\cite{Pritchard-ESA09}, which was improved to an $O(k^2)$-approximation by Chekuri, Ene, and Korula (see~\cite{ChakPritch-Algorithmica10} and~\cite{BansalEtAl-IPCO10}) and Chakrabarty and Pritchard~\cite{ChakPritch-Algorithmica10}.  Most recently, Bansal \emph{et al.}~\cite{BansalEtAl-IPCO10} devised a deterministic $8k$-approximation and a randomized $(ek+o(k))$-approximation.     

\paragraph{\textbf{Outline.}} In the following section we motivate iterative packing with an example.  In Sect.~\ref{sec:k-HDM} we apply the method to design a $2k$-approximation for $k$-HDM.  We finally present a more refined application in Sect.~\ref{sec:2-CS-PIP} to derive a 3-approximation for $2$-CS-PIP.

\section{Iterative Packing: An Example}\label{sec:iterative-packing}
We illustrate iterative packing on the maximum matching problem.  Although this is a simple application, it serves well to illustrate the method.  Consider the natural degree-based LP relaxation $P_{M}(G)$ for the maximum matching problem on a graph $G=(V,E)$: 
\begin{alignat}{3}
\text{Maximize} &\quad& \sum_{e \in E} c_e x_e &
  \tag*{$P_{M}(G)$}\label{lp:PM}\\ \text{subject to} &&
  \sum_{e \in \delta(v)} x_e &\leq 1 &\quad& \forall v \in
  V\label{ineq:PM:pack}\\ && 0 \leq x_e & \leq 1 && \forall e \in
  E. \label{ineq:PM:nonneg}
\end{alignat}
Given a feasible fractional solution $x^* \in P_{M}(G)$, the iterative packing procedure obtains an \emph{$\alpha$-approximate} convex decomposition of $x^*$, \begin{equation}\label{eq:pack:convex-decomp}
\alpha x^* = \sum_{i \in \mathcal{I}} \lambda_{i} \chi^{i}\enspace,
\end{equation}
for some $\alpha \in (0,1]$, where each $\chi^{i} \in P_{M}(G)$ is an integral solution (and $\sum_{i} \lambda_{i} = 1$ and $\lambda_{i} \geq 0$ for all $i$).  Iterative packing in its most basic form directly produces a sparse decomposition, namely one with $|\mathcal{I}| \leq |E|+1$.  Even when this is not the case, we can apply elementary linear algebra to retain at most $|E|+1$ solutions (more generally $n+1$, where $x^* \in \bbbr^n$).  A procedure to accomplish the latter is related to Carath{\'e}odory's Theorem and makes for a good exercise.

The construction of the decomposition \eqref{eq:pack:convex-decomp} implies that one can find an integral solution with cost at least $\alpha(c x^*)$, thus $1/\alpha$ corresponds to the approximation guarantee of the resulting approximation algorithm. A nice feature is that the decomposition gives us a cost oblivious representation of an approximate solution.

For $P_{M}(G)$, we first show that choosing $\alpha = 1/2$ suffices.  This yields a $2$-approximation while also showing that the integrality gap of $P_{M}(G)$ is at most 2.  We then show that we may select $\alpha=2/3$ by leveraging the fact that extreme points of $P_{M}(G)$ must contain an edge $e$ with $x_e \geq 1/2$ (in fact this holds for all $e$).  The latter precisely matches the integrality gap of $P_{M}(G)$. This is interesting, since much like iterative rounding, iterative packing offers insight into how large fractional components can facilitate the approximation of packing problems.  Akin to iterative rounding, iterative packing is motivated by a simple idea.  We start with a fractional solution $x^*$ and:
\begin{enumerate}
\item Remove an edge $e$ (without otherwise modifying the instance) 
\item Recursively obtain an $\alpha$-approximate convex decomposition of the resulting fractional solution, $\bar{x}^*$
\item Pack $e$ into precisely an $\alpha x^*_{e}$ fraction of the integral solutions. 
\end{enumerate}

The key, of course, is showing that the last step can always be performed successfully. For this to work, we require that for any fractional (or perhaps extreme point) solution $x^*$ there exists an $e \in E$ with
\begin{equation}\label{eq:pack:insert-e}
\sum_{i \in \mathcal{I}_{e}} \lambda_{i} \geq \alpha x^*_e\enspace,
\end{equation}
where $\alpha \bar{x}^* = \sum_{i \in \mathcal{I}} \lambda_{i} \chi^{i}$ is an arbitrary approximate convex decomposition of the residual solution, $\bar{x}^*$, and $i \in \mathcal{I}_{e}$ indicates that $\chi^i$ is able to accommodate the edge $e$ (i.e. $\chi^i \cup e$ is a matching).

Although we may well be able to pack $e$ into a fraction of the integral solutions that is larger than an $\alpha x^*_e$, to maintain our implicit inductive hypothesis we must ensure that $e$ is packed into exactly an $\alpha x^*_{e}$ fraction of solutions.  To accomplish this, we may have to clone some solution $\chi^{i}$, insert $e$ into exactly one of the two copies of $\chi^{i}$, and distribute the multiplier $\lambda_{i}$ among the copies so that $e$ appears in the requisite fraction of solutions.  The base case, which contains no edges, selects the empty solution with a multiplier of 1.  Thus if \eqref{eq:pack:insert-e} holds universally for a particular value of $\alpha$, then we can efficiently obtain an $\alpha$-approximate convex decomposition of $x^*$ consisting of at most $|E|+1$ integral solutions.  Selecting the best of these gives us the corresponding approximation algorithm.

To see that \eqref{eq:pack:insert-e} holds when $\alpha = 1/2$, consider some fractional solution $x^*$ and an arbitrary edge $e = uv \in E$ with $x^*_{e} > 0$.  Obtaining $\bar{x}^*$ as above by deleting $e$, we have that 
$$\max\{\bar{x}^*(\delta(u)),\,\bar{x}^*(\delta(v))\} \leq 1-x^*_{e}\enspace,$$ 
hence in any convex decomposition $\alpha\bar{x}^*$, at most a $2\alpha(1-x^*_{e})$ fraction of the $\chi^{i}$ do not accomodate $e$, hence we require $1-2\alpha(1-x^*_{e})  \geq  \alpha {x^*}_{e}$,
which is equivalent to 
\begin{equation}\label{eq:pack:apx}
\alpha \leq \frac{1}{2 - x^*_e}
\end{equation}
Thus by selecting $\alpha = 1/2$, we may successfully pack any edge $0 \leq x^*_e \leq 1$ in the last step of our algorithm.  However, by selecting a large edge at each iteration we can improve the bound.  It is well known that extreme points of $P_{M}(G)$ are $1/2$-integral, so we may actually take $\alpha=1/(2-1/2)=2/3$.  More generally -- just as with iterative rounding -- it suffices to show
that an extreme point always contains some edge of large fractional value.  We explore this idea in conjunction with $2$-CS-PIP in Sect.~\ref{sec:2-CS-PIP}.  However, in the next section we show that the framework above with a simple modification yields a $2k$-approximation for $k$-HDM.

\section{Iterative Packing for $k$-Hypergraph Demand Matching}\label{sec:k-HDM}

The results in this section are obtained using the framework outlined for our the matching problem in the previous section:
\begin{enumerate}
\item Remove a (hyper)edge $S$ (without otherwise modifying the instance) 
\item Recursively obtain an $\alpha$-approximate convex decomposition of the resulting fractional solution, $\bar{x}^*$
\item Pack $e$ into precisely an $\alpha x^*_{S}$ fraction of the integral solutions. 
\end{enumerate}
However, directly applying the algorithm above to $k$-HDM does not yield a bounded approximation guarantee (w.r.t. $k$).  We show that simply selecting an edge $S$ with minimal demand in step 1 above yields a $2k$-approximation.

As with our analysis in the previous section, the crux lies in being able to carry out step 3 successfully.  Following the analysis in Sect.~\ref{sec:iterative-packing}, let  $\alpha \bar{x}^* = \sum_{i \in \mathcal{I}} \mu_{i} \chi^{i}$ be an arbitrary convex decomposition of the residual solution $\bar{x}^*$ obtained in step 2.  We may inductively assume the existence of such a decomposition (with a trivial base case).  To determine whether the edge $S$ may be packed in the requisite $\alpha x^*_S$ fraction of the integral solutions $\chi^i$, we consider the potential fraction of solutions in which $S$ is blocked at each $u \in S$.  For $u \in S$, let $\beta_u$ be the fraction of solutions in which $S$ cannot be packed at $u$.  We may think of $\beta_u$ as the fraction of bad solutions in terms of packing $S$.  In the worst case, $S$ is blocked pairwise disjointly at each incident vertex.

\begin{lemma}\label{lem:apx-ineq} Edge S may be packed into an $\alpha x^*_S$ fraction of the integral solutions $\chi^i$, provided
$$1 - \sum_{u \in E} \beta_u \geq \alpha x^*_S\enspace.$$
\end{lemma}
\begin{proof}
The lemma follows from the same reasoning used to derive a union bound if $x^*$ were a probability distribution.  The quantity $\sum_{u \in S} \beta_u$ represents the maximum fraction of solutions in which $S$ is blocked at some incident vertex, hence $1 - \sum_{u \in S} \beta_u$ is the minimum fraction of solutions into which it is feasible to insert $S$.
\qed\end{proof}

We may derive a $1/\alpha$-approximation guarantee on the performance of iterative packing by bounding $\beta_u$ and selecting $\alpha$ so that Lemma~\ref{lem:apx-ineq} is satisfied.  For this purpose we find it useful to think of the residual convex decomposition, $\alpha \bar{x}^* = \sum_{i \in \mathcal{I}} \mu_i \chi^i$, obtained in step 2. above, as inducing a collection of bins at each $u \in V$ where each bin has capacity $b_u$.  Each $i \in \mathcal{I}$ induces a bin of width $\mu_i$; the height $h_i$ is equal to the sum of the demands of the edges incident upon $u$ that appear in the solution $\chi^i$; that is 
$$h_i := \sum_{S \in\chi^i | u \in S} d_S\enspace.$$
Thus the aggregate capacity of the bins is at most $(\sum_{i \in \mathcal{I}} \mu_i) b_u = b_u$, where each bin contains a volume of $\mu_i h_i$.  

Next we bound the fraction of bad solutions at $u$.  For this, we define $\bar{\delta} = \min_{T \in \mathcal{E} | T \not=S} d_T$, i.e. $\bar{\delta}$ is the minimal demand in the residual instance.

\begin{lemma}\label{lem:packing-bound}
For the convex decomposition, 
$\alpha \bar{x}^* = \sum_{i \in \mathcal{I}} \mu_i \chi^i$, we have
$$\beta_u \leq \alpha\frac{b_u - d_S x^*_S}{\max\{b_u - d_S + 1,\, \bar{\delta}\}}
\enspace.$$
\end{lemma}
\begin{proof}
Let $\mathcal{I}_{\overline{S}}$ be the indices of the bins which cannot accommodate $S$.  Thus by definition,
$$\beta_u = \sum_{i \in \mathcal{I}_{\overline{S}}} \mu_i\enspace.$$
The total volume of all such bins is at most the total $u$-volume of $\alpha \bar{x}$, which does not contain $S$:
$$\sum_{i \in \mathcal{I}_{\overline{S}}} \mu_i h_i \leq \alpha(b_u - d_S x^*_S)\enspace.$$
Each bin in $\mathcal{I}_{\overline{S}}$ must have height large enough to block $S$ and must also contain at least one edge since $S$ fits on its own, by the no clipping assumption.  Thus $h_i \geq \max\{b_u - d_S + 1, \,\bar{\delta}\},$ yielding the desired result when coupled with the above equation and inequality:
$$\max\{b_u - d_S + 1, \,\bar{\delta}\}\beta_u = \max\{b_u - d_S + 1, \,\bar{\delta}\}\sum_{i \in \mathcal{I}_{\overline{S}}} \mu_i \leq \alpha(b_u - d_S x^*_S)\enspace.$$
\qed\end{proof}

Unfortunately when $d_S$ is large and $x^*_S$ is small, the above bound may be large.  However, by appealing to special cases of $k$-HDM or by more carefully selecting the edge $S$, the bound $\beta_u$ becomes manageable.  For instance, consider the case of the $k$-hypergraph $b$-matching problem, obtained when $d_S = 1$ for all $S$.  In this case $\beta_u \leq \alpha$ by Lemma~\ref{lem:packing-bound}, which allows us to satisfy the hypothesis of Lemma \ref{lem:apx-ineq} by selecting $\alpha$ such that:
\begin{equation}\label{eq:gamma-leq-alpha}
\alpha \leq \frac{1}{x^*_S + k} \Rightarrow \alpha x^*_S \leq 1-k\alpha \leq 1 - \sum_{u \in E} \beta_u\enspace.
\end{equation}
Since $x^*_S \leq 1$ for all $E$, we may universally select $\alpha = \frac{1}{k+1}$ for all $S$, yielding an approximation guarantee of $k+1$.  Krysta~\cite{Krysta-MFCS05} proved that a greedy algorithm also achieves this bound, and Young and Koufogiannakis~\cite{YoungKouf-DISC09} give a primal dual algorithm achieving a bound of $k$, which is the best known.  Although we omit the details in this article, iterative packing can be used to show that  the integrality gap of the natural LP relaxation for $k$-hypergraph $b$-matching is at most $k-1+1/k$, which settles the gap.

Turning our attention back to $k$-HDM, the ``max'' in Lemma $\ref{lem:packing-bound}$'s bound hints at our strategy: we shall always select an edge $S$ with minimal demand, so that $\bar{\delta}$ is large enough to be of value.  In fact the resulting approximation algorithm applies to a generalization of $k$-HDM ($k$-CS-CPIP) in which we allow each edge $S$ to have a different valued demand, $d^S_v$ at each vertex $v$, as is allowed in $k$-CS-PIP.  However, we require that the edges can be ordered, $S_1, S_2, \ldots, S_m$, so that for any distinct $S_i, S_j$ with $u \in S_i \cap S_j$, we have $d^{S_i}_u \leq d^{S_j}_u$ if $i \leq j$; that is, the demands monotonically increase at every vertex.  Note that this is clearly the case with $k$-HDM, where $d_S = d^S_u = d^S_v$ for all $u,v \in S$.  We may simply sort the demands over the edges to obtain such an ordering.  We perform iterative packing with such an ordering (i.e. select an S with minimal demand).  Now when we insert $S$ back into the approximate convex decomposition of $\alpha \bar{x}$, we may assume that $d_S \leq \bar{\delta}$. 
\begin{theorem}Iterative packing applied to $k$-CS CPIP with the edges inserted in order of nonincreasing demand is a $2k$-approximation.
\end{theorem}

\begin{proof}   To simplify our analysis, we conduct it in terms of $1/\alpha$.  The stipulation of Lemma~\ref{lem:apx-ineq} may be expressed as:
$$x^*_S + \sum_{u \in S} \frac{\beta_u}{\alpha} \leq \frac{1}{\alpha}\enspace.$$ 
Our goal is to show that the above holds when $1/\alpha=2k$.  By applying the bound on $\beta_u$ from Lemma~\ref{lem:packing-bound}, we reduce our task to showing that
$$x^*_S + \sum_{u \in S} \frac{b_u - d_S x^*_S}{\max\{b_u - d_S + 1,\, \bar{\delta}\}} \leq 2k\enspace,$$ 
for any value of $x^*_S \in [0,1]$.  Note that when the left hand side above is considered as a function of the parameters $b_u$, $d_S$, and $x^*_S$, it is linear in $x^*_S$.  Thus it is maximized in one of the cases $x^*_S = 0$ or $x^*_S=1$.  When $x^*_S = 1$ we indeed have
$$1 + \sum_{u \in S} \frac{b_u - d_S}{\max\{b_u - d_S + 1,\, \bar{\delta}\}} \leq  1 + \sum_{u \in S} \frac{b_u - d_S}{b_u - d_S + 1} \leq 1 + k \leq 2k\enspace.$$
On the other hand when $x^*_S = 0$, we have
$$0 + \sum_{u \in S} \frac{b_u}{\max\{b_u - d_S + 1,\, \bar{\delta}\}} \leq
  2\sum_{u \in S} \frac{b_u}{b_u + \bar{\delta} - d_S + 1} \leq 2k\enspace,$$
where the former inequality follows because $\max\{x,y\} \geq (x+y)/2$, and the latter holds because our ordering of the edges gives us $\bar{\delta} \geq d_S$. 
\qed\end{proof}

\paragraph{\textbf{Integrality gap.}} Our result essentially settles the integrality gap of the natural formulation.  As noted in~\cite{KahnEtAl-Combinatorica93} the projective plane of order $k-1$ yields an integrality gap of at least $k-1+1/k$, even for the case of $k$-hypergraph matching.  For $k$-HDM (and consequently $k$-CS PIP) one may obtain a lower bound approaching $2(k-1+1/k)$ by again considering a projective plane of order $k-1$, setting all the demands to $d$ and each capacity to $b = 2d-1$.  

\section{Improvements for 2-CS-PIP and Demand Matching}\label{sec:2-CS-PIP}

Here we consider the general $k$-CS-PIP rather than $k$-HDM/$k$-CS-CPIP, but for the special case $k=2$.  This case is of particular interest as it is natural generalization of the demand matching problem (i.e. $2$-CS-CPIP), which itself is combination of both $b$-matching and knapsack type problems in graphs.

Shepherd and Vetta~\cite{ShepVet-MOR07} were able to show that integrality gap of the natural LP formulation was between 3 and 3.264; however, establishing the exact value has remained an open problem prior to our work.  We are able to show that there is indeed a 3-approximation based on the natural LP for not only demand matching but also the more general $2$-CS-PIP problem.  This consequently settles the integrality gaps for both problems.  Although designing a polynomial time algorithm takes a bit of work, iterative packing allows us to establish the integrality gap relatively easily.

\subsection{Establishing the Integrality Gap}

As observed in the introduction for the standard matching problem, iterative packing readily yields an upper bound of 2 on the integrality gap of the natural formulation, which is sharpened to the optimal value of 3/2 (= $k-1+1/k$) by observing that extreme points must contain some large component ($x_e \geq 1/2$) and iterating only on such edges.  For $2$-CS PIP we also apply iterative packing solely on large components in extreme points -- in fact, those with $x_e = 1$ when they exist.  

An interesting phenomenon with general demand packing problems is that $1$-edges (i.e. $x_e=1$) cannot simply swept under the rug as with $\{0,1\}$-demand problems.  For the latter, one can simply assume such edges are selected in a solution and obtain a straightforward residual instance and feasible fractional solution.  With general demand problems, however, such an approach may violate the no clipping assumption.  Iterative packing, on the other hand, performs quite well on $1$-edges, which allows us to rather easily prove an optimal integrality gap.  Although we will also derive this result in the next section by means of a polynomial time algorithm, in this section we attack only the integrality gap with a short proof that highlights the effectiveness of iterative packing when augmented with simple insights about the nature of extreme points.  Our first task is to show that extreme points without $1$-edges admit a manageable structure.  We note that since our discussion involves the more general 2-CS-PIP rather than demand matching, each edge $uv$ may have distinct demands $d^{uv}_u$ and $d^{uv}_v$.  

\begin{lemma}\label{lem:2-CS-PIP-ext-point}
If $\hat{x}$ is an extreme point of the natural LP then the fractional part of $\hat{x}$ induces connected components with at most one cycle.  If we have $0 < \hat{x} < 1$, then $\hat{x}$ induces vertex disjoint cycles.
\end{lemma}
\begin{proof} Let $F \subseteq E$ be the fractional support of $\hat{x}$.  For each connected component $C$ induced by $F$,  the only tight constraint in which $e \in F(C)$ may appear are degree constraints,
 $$ \forall u \in V: \sum_{uv \in \delta(u)} d^{uv}_u x_{uv} \leq b_u\enspace, $$
for $u \in V(C)$.  Thus $|F(C)| \leq |V(C)|$, otherwise we would find that a basis for $\hat{x}$ contains linearly dependent columns among those of $F(C)$.  This establishes the first claim of the lemma.

If $0 < \hat{x} < 1$, then no component induced by $\hat{x}$ may contain an edge $e$ incident to a leaf $l$, otherwise we would have $d^e_l \hat{x}_e = b_l$, implying $\hat{x}_e = 1$ by the no clipping assumption.  Thus we must have $|F(C)| = |V(C)|$ for each component $C$, and since we have no leaves, $C$ must be a cycle.
\qed\end{proof}

Coupled with our earlier analysis of iterative packing, this is the only fact we need to establish the integrality gap of 3.  Consider the following non-efficient extension of iterative packing:

\begin{enumerate}
\item If $x^*$ is not an extreme point, obtain a convex decomposition into extreme points, $x^* = \sum_i \mu_i \hat{x}^i$, and apply the algorithm to each extreme point $\hat{x}^i$.
\item If the extreme point $\hat{x}$ contains an integral edge let $e$ be such an edge, otherwise let $e$ be any edge.
\item Delete $e$ to obtain $\bar{x}$ and recursively construct an approximate convex decomposition into integral solutions, $\frac{1}{3}\bar{x} = \sum_j \lambda_j \chi^j$.
\item Insert $e$ into exactly a $\frac{1}{3}x_e$ fraction of the solutions $\chi^j$.
\end{enumerate}

\begin{lemma}\label{lem:pack-1-edges}
Step \mbox{4.} above can always be completed successfully.
\end{lemma}
\begin{proof} Suppose there is an integral $\hat{x}_e$ and that $\hat{x}_e = 1$ ($\hat{x}_e = 0$ clearly works). Substituting $\hat{x}_e = 1$ into the bound from Lemma~\ref{lem:packing-bound} yields
$\beta_u \leq \alpha$ for $u \in e$, which we have already observed (see~\eqref{eq:gamma-leq-alpha}) allows us to select $\alpha = 1/(k+1) = 1/3$.  On the other hand, if there is no integral edge by Lemma~\ref{lem:2-CS-PIP-ext-point}, $\hat{x}$ induces a 2-regular graph.  When we delete $e$ in this case, at each endpoint $u \in e$ there is a single edge, $f_u$ remaining.  Thus $\beta_u \leq \alpha \hat{x}_{f_u} \leq \alpha$ in this case as well.
\qed\end{proof}

We have a lower bound on the integrality gap of $2(k-1+1/k) = 3$ from the previous section.  To complete our proof, we note that by assuming an approximate convex decomposition for each  extreme point, $\frac{1}{3}\hat{x}^i = \sum_j \lambda_j \chi^{ij}$, we may obtain a decomposition for $\frac{1}{3}x^*$ as $\sum_i \sum_j \mu_i \lambda_j \chi^{ij}$.

\begin{theorem}The integrality gap of the natural LP formulation for $2$-CS PIP is 3.
\end{theorem}

\subsection{A Polynomial Time Implementation}

Unfortunately we do not currently have a means of directly implementing the algorithm above in polynomial time; however,  Shepherd and Vetta~\cite{ShepVet-MOR07} analyze the fractional structure of extreme points of the natural LP for demand matching.  We are able to develop a polynomial time algorithm by relying on generalizations of their demand matching insights.  A key ingredient used by Shepherd and Vetta is the augmentation of a fractional path (Sect. 4.1 in~\cite{ShepVet-MOR07}). We begin by giving a generalization of this tool for the case of $2$-CS PIP.

\begin{lemma}\label{lem:augmentation}
Let $P = (v_0, e_1, v_1, e_2, \ldots, v_k)$ be a path; there exists an augmentation vector $z \in \bbbr^E$ such that
\begin{enumerate}
\item $\sum_{uv \in \delta(u)} d^{uv}_u z_{uv} \not= 0,$ for $u = v_0, v_k$ 
\item $\sum_{uv \in \delta(u)} d^{uv}_u z_{uv} = 0,$ for all other $u \in V$
\end{enumerate}
\end{lemma}
\begin{proof}  We set $z_e = 0$ for all $e \notin E(P)$, and we set $z_{e_1} = 1$.  We now set the value of $z_{e_{i+1}}$, for $i \geq 1$, based on the value of $z_{e_i}$ as follows:
$$ z_{e_{i+1}} = -(d^{e_i}_{v_i}/d^{e_{i+1}}_{v_i}) z_{e_i}\enspace.$$
The first condition is satisfied since $z_{e_i} \not = 0 \Rightarrow z_{e_{i+1}} \not= 0$, and the second condition holds since $d^{e_{i+1}}_{v_i} z_{e_{i+1}} = -d^{e_i}_{v_i} z_{e_i}$.
\qed\end{proof}

\paragraph{\textbf{Algorithm.}} We will explain the utility of the above lemma in just a moment; however, first we give an overview of our algorithm:
\begin{enumerate}
\item Find an extreme point $\hat{x}$ of the natural $2$-CS PIP LP.
\item Delete any $0$-edges and iterate on the $1$-edges until a complete fractional solution $\bar{x}$ remains.
\item We will show that $\bar{x}$ possesses a structure that allows us to infer a $3$-approximation algorithm based on a result of Shepherd and Vetta.
\item Apply a result of Carr and Vempala~\cite{CarrVempala-RSA02} with the above $3$-approximation algorithm as an approximate separation oracle to obtain an approximate convex decomposition of $\bar{x}$ in polynomial time.
\item Pack the removed $1$-edges into $1/3$ of the solutions from the above decomposition.
\end{enumerate}

The basic idea is to use a more complex base case for iterative packing, rather than the default case of an empty graph.  In the above algorithm steps 3 and 4 represent the base case.  We address the results employed in these steps in turn.

\paragraph{\textbf{Analysis.}} First we describe the 3-approximation derived from Shepherd and Vetta's work.  Note that in step 3, we have a solution $\bar{x}$ that contains precisely the fractional components of an extreme point $\hat{x}$.  By Lemma~\ref{lem:2-CS-PIP-ext-point}, each component induced by $\bar{x}$ is either a tree or unicyclic.  For the former case, we can apply a generalization of Thereom 4.1 from Shepherd and Vetta~\cite{ShepVet-MOR07}, which yields a 2-approximation with respect to a fractional solution $x^*$ whose support is a tree.  They use path augmentations to derive this result for demand matching, for which we give an appropriate generalization in Lemma~\ref{lem:augmentation}.

\paragraph{\textbf{A 3-approximation.}} We briefly explain the basic idea behind the 2-approximation mentioned above and recommend the reader consult Sect. 4.1 in~\cite{ShepVet-MOR07} for a rigorous proof.  Since $x^*$ induces a tree, we can find a path $P$ between two leaves $s$ and $t$.  We apply Lemma~\ref{lem:augmentation} on $P$ to obtain $z$; we are able to select an $\varepsilon \not= 0$ such that:
\begin{itemize}
\item $x^* + \varepsilon z$ earns cost at least that of $x^*$
\item $x^* + \varepsilon z$ is still feasible at every vertex except possibly $s$ and $t$
\item $x^* + \varepsilon z$ has some integral edge
\end{itemize} 
Thus we obtain a new solution of no worse cost, but it may be infeasible at $s$ and $t$.  We temporarily remove any integral edges to obtain smaller trees and continue the procedure to obtain a collection of integral edges that are not necessarily feasible but are of superoptimal cost.  The key observation is that when a vertex becomes infeasible, it is a leaf, which allows Shepherd and Vetta to conclude in the final solution, at every vertex there is at most one edge whose removal results in a feasible solution.  Since the edges form a forest, Shepherd and Vetta are able to partition the edges into two sets such that each is feasible, yielding a 2-approximation.

Returning to our algorithm, each component of $\hat{x}$ contains at most one cycle, thus for a given cost function, we may delete the cheapest edge from each such cycle to retain a solution of cost at least $2/3$ that of $\hat{x}$.  This leaves a forest on which we apply Shepherd and Vetta's procedure to obtain an integral solution of cost at least $1/3 = 1/2\cdot 2/3$ that of $\hat{x}$.  The trouble is that although this gives a $3$-approximation, we actually need a convex decomposition of $1/3\hat{x}$ in order to pack the $1$-edges removed in the second step.  Luckily, we are able to appeal to the following result Carr and Vempala~\cite{CarrVempala-RSA02} to obtain such a decmoposition.

\begin{theorem}(Thm 2,~\cite{CarrVempala-RSA02}) Given an LP relaxation P, an r-approximation heuristic A, and any solution $x^*$ of P, there is a polytime algorithm that finds a polynomial number of integral solutions $z^1, z^2, \ldots$ of P such that
$$\frac{1}{r}x^* \leq \sum_j \lambda_j \chi^{z^j}$$
where $\lambda_j \geq 0$ for all $j$, and $\sum_j \lambda_j = 1$.
\end{theorem}

The theorem from Carr and Vempala is for covering problems; however, their ideas yield the result for the packing case as well.  We also note that the heuristic $A$ must be an $r$-approximation with respect to the lower bound given by the relaxation $P$.

\paragraph{\textbf{Obtaining an exact decomposition.}} The only remaining detail is that  for our purposes, we require an exact decomposition of $x^*/r$ (i.e. $x^*/r = \sum_j \lambda_j \chi^{z^j}$).  We observe that this can be done by at most doubling the number integral solutions $z^j$.  For any edge $e$ such that $x_e^*/r > \sum_j \lambda_j \chi^{z^j}_e$, we remove $e$ from solutions that contain $e$.  We continue this until removing $e$ from any solution would give us $x_e^*/r \leq \sum_j \lambda_j \chi^{z^j}_e$.  Now we clone some solution $z^j$ that contains $e$ and remove $e$ from the clone; we distribute $\lambda_j$ between these two solutions to attain $x_e^*/r = \sum_j \lambda_j \chi^{z^j}_e$.  Finally, Lemma~\ref{lem:pack-1-edges} shows that step 5. may be completed.

Carr and Vempala's result is essentially an application of LP duality and the polynomial time equivalence of separation and optimization.  It certainly seems likely that a more direct analysis of $1/3\hat{x}$ could be used to construct a convex decomposition; however, we are enamored with the elegance of Carr and Vempala's result.

\section{Concluding Remarks}

We have obtained simple iterative packing algorithms for $k$-hypergraph demand matching problems that essentially settle the gap of the respective natural LP relaxations.  Obvious open questions include generalizing our work to $k$-CS-PIP.  We are currently investigating this and have derived promising partial results.   It is conceivable that one may also develop an analogue of iterative packing for covering that appeals to approximate convex decompositions.

\bibliographystyle{splncs03}
\bibliography{iterative-packing}

\end{document}